\documentclass{article}
\usepackage[letterpaper]{geometry}
\usepackage[utf8]{inputenc}
\usepackage[T1]{fontenc}
\usepackage{url}
\usepackage{ifthen}
\usepackage[cmex10]{amsmath}
\def\withcolors{0} \def\withnotes{0} 

\usepackage{varwidth}
\usepackage[utf8]{inputenc} 
\usepackage[T1]{fontenc} 
\usepackage{hyperref}
\usepackage{url} 
\usepackage{booktabs} 
\usepackage{amsfonts} 
\usepackage{nicefrac}
\usepackage{microtype} 
\usepackage[noend]{algpseudocode}
\usepackage{soul}
\usepackage[T1]{fontenc}
\usepackage{lmodern}

\usepackage[usenames,dvipsnames,table]{xcolor}
\usepackage{amsmath}
\usepackage{algorithm}
\usepackage[noend]{algpseudocode}
\algsetblock[Name]{Encoder}{Stop}{3}{1cm}
\algsetblock[Name]{Decoder}{Stop}{3}{1cm}
\algsetblock[Name]{SharedRandomness}{Stop}{3}{1cm}

\usepackage[colorinlistoftodos,textsize=scriptsize]{todonotes}
\usepackage[mathscr]{eucal}

\usepackage[noadjust]{cite}
\usepackage{subcaption}
\usepackage{%
    amssymb,%
    amsthm,
    bbm,%
    etoolbox,%
    mathtools,%
    stmaryrd%
}

\usepackage{amsfonts}
\usepackage{amssymb}
\usepackage{enumitem}
\newlist{steps}{enumerate}{1}
\setlist[steps, 1]{label = \quad \quad \quad Step \arabic*:}
\usepackage{tikz,pgf,pgfplots,circuitikz}
\usepgflibrary{shapes}
\usetikzlibrary{%
    arrows,%
    decorations.text,%
    positioning,%
    scopes,%
    shapes,
    patterns%
}

\theoremstyle{plain} \newtheorem{thm}{Theorem}[section]
\newtheorem{lem}[thm]{Lemma}

\theoremstyle{definition} \newtheorem{defn}[thm]{Definition}

\theoremstyle{remark} \newtheorem{rem}{Remark}

\definecolor{lightgray}{gray}{0.9}

\newcommand{\ep}{\mathcal{E}}
\newcommand{\eqdef}{:=}
\newcommand{\eq}[1]{\begin{align*}#1\end{align*}}
\newcommand{\ceil}[1]{\left\lceil#1\right\rceil}%
\newcommand{\norm}[1]{\|#1\|}%

 \newcommand{\R}{\mathbb{R}}
\newcommand{\N}{\mathbb{N}}
\newcommand{\E}[1]{\mathbb{E}\left[#1\right]}

\newcommand{\V}{\mathcal{V}}
\newcommand{\X}{\mathcal{X}}

\newcommand{\oO}{\mathcal{O}}
\newcommand{\Q}{\mathcal{Q}}
\newcommand{\uRoman}[1]{\uppercase\expandafter{\romannumeral#1}}

\usepackage{subcaption}
\DeclareCaptionSubType*{algorithm}

\DeclareCaptionLabelFormat{alglabel}{Alg.~#2}

\ifnum\withcolors=1
  \newcommand{\newer}[1]{{\color{brown} {#1}}} 
  \newcommand{\newest}[1]{{\color{blue} {#1}}} 
  \newcommand{\mcolor}[1]{{\color{ForestGreen}#1}} 
  \newcommand{\tcolor}[1]{{\color{Orange}#1}} 
\else
  
  \newcommand{\newer}[1]{{{#1}}}
  \newcommand{\newest}[1]{{{#1}}}
  \newcommand{\mcolor}[1]{{#1}}
  
  \newcommand{\tcolor}[1]{{#1}}
\fi

\ifnum\withnotes=1
  \newcommand{\mnote}[1]{\par\mcolor{\textbf{M: }\sf #1}} 
  \newcommand{\tnote}[1]{\par\tcolor{\textbf{T: }\sf #1}} 
  
\else
  \newcommand{\mnote}[1]{}
  \newcommand{\tnote}[1]{}
  
\fi
\newcommand{\ignore}[1]{\leavevmode\unskip} 

\newcommand{\Qenc}{Q^{{\tt e}}}
\newcommand{\Qdec}{Q^{{\tt d}}}

\newcommand{\indic}[1]{\mathbbm{1}_{#1}}

\input{tangocolors.sty}
\usepackage{framed}
\usepackage{xcolor}
\usepackage{colortbl}
\usepackage[utf8]{inputenc} 
\usepackage{hyperref} 
\usepackage{url} 
\usepackage{booktabs} 
\usepackage{amsfonts} 
\usepackage{nicefrac} 
\usepackage{microtype} 
\usepackage[noend]{algpseudocode}
\usepackage{soul}
\usepackage{tcolorbox}

\begin{document}
\title{Limits on Gradient Compression for Stochastic Optimization}
\author{Prathamesh Mayekar$^\dag$ \and Himanshu Tyagi$^\dag$ }
\maketitle

{\renewcommand{\thefootnote}{}\footnotetext{
\noindent$^\dag$Department of Electrical Communication Engineering,
Indian Institute of Science, India.  Email: \{prathamesh, htyagi\}@iisc.ac.in. }}
\begin{abstract}
  We consider stochastic optimization over $\ell_p$ spaces using
  access to a first-order oracle.  We ask: {What is the minimum
    precision required for oracle outputs to retain the unrestricted
    convergence rates?}  We characterize this precision for every
  $p\geq 1$ by deriving information theoretic lower bounds and by
  providing quantizers that (almost) achieve these lower bounds.  Our
  quantizers are new and easy to implement.  In particular, our
  results are exact for $p=2$ and $p=\infty$, showing the minimum
  precision needed in these settings are $\Theta(d)$ and $\Theta(\log
  d)$, respectively. The latter result is surprising since recovering
  the gradient vector will require $\Omega(d)$ bits.
\end{abstract}

\section{Introduction}

We consider the classic problem of minimizing an unknown convex
functions which \newer{is} Lipschitz continuous in the $\ell_p$ norm.
Our setting is that of first-order stochastic optimization, where we
are given oracle access to noisy, unbiased estimates of the
subgradients which have their $\ell_q$ norm bounded, where $q$ is the
H\"{o}lder conjugate of $p$. Motivated by recent works on gradient
descent using quantized gradient updates
($cf.$~\cite{de2015taming,gupta2015deep, suresh2017distributed,
  wen2017terngrad, alistarh2017qsgd, stich2018sparsified,
  wang2018atomo, agarwal2018cpsgd, acharya2019distributed,
  karimireddy2019error, basu2019qsparse, ramezani2019nuqsgd,
  gandikota2019vqsgd, mayekarratq}), we seek to determine the minimum
number of bits $r$ to which the gradients can be expressed while
retaining the standard, unrestricted gradient rates.

We show that for $p\in [1,2)$ and $p\geq 2$, respectively, roughly $d$
  and $d^{2/p}\log (d^{1-2/p}+1)$ bits are necessary and sufficient
  for retaining the standard convergence rates. These bounds are tight
  upto an $O(\log d)$ factor, in general, but are exact for $p=2$ and
  $p=\infty$. Prior work has only considered the problem for the
  Euclidean case, and not for general $\ell_p$ geometry. Further, even
  for the Euclidean case, the best known bounds are from our recent
  work~\cite{mayekarratq} where the bounds are shown to be tight only
  upto a mild $O(\log \log \log \ln^\ast d)$ factor. The results in
  this paper get rid of this nagging factor and establish tight
  bounds.

We use different quantizers for $p\geq 2$ and $p\in [1,2)$.  In the
  $p\geq 2$ range, we use a quantizer we call $SimQ^+$.  $SimQ^+$, in
  turn, uses multiple repetitions of another quantizer we call $SimQ$
  which expresses a vector as a convex combination of corner points of
  an $\ell_1$ ball. It is $SimQ$ that yields an $O(\log d)$ bit
  quantizer for optimization over $\ell_\infty$.  Also, $SimQ^+$
  yields the exact upper bound in the $\ell_2$ case.  In the $[1,2)$
    range, we divide the vector into two parts with small and large
    coordinates. We use a uniform quantizer for the first part and
    RATQ of~\cite{mayekarratq} for the second part.

The main observation in our analysis for upper bound is that the role
of quantizer in optimization is not to express the gradient with small
error. It suffices to have an unbiased estimate with appropriately
bounded norms. Our lower bounds are based on those
in~\cite{agarwal2012information,mayekarratq}. But interestingly we
show that bounds that are useless in the classic setting become
useful under precision constraints.

We remark that while our quantizers are related to the ones
  used in prior works, our main contribution is to show that our
  specific design choices yield optimal precision. For instance, the
  quantizers in~\cite{gandikota2019vqsgd} expresses the input as a
  convex combination of set of points, similar to $SimQ$.  In fact,
  one of the quantizers in \cite{gandikota2019vqsgd} uses similar set
  of points as that of SimQ with a different
  scaling. However, the quantizers in \cite{gandikota2019vqsgd} are
  designed keeping in mind other objectives and they fall short of
  attaining the optimal precision guarantees of $SimQ$ and $SimQ^+$.
  
\newest{Also, stochastic optimization over $\ell_p$ spaces using  a biased first-order oracle which is constructed by using only statistical query access to the underlying data
  was considered in~\cite{feldman2017statistical}.} Here, the authors characterize the number of statistical queries needed to optimize a function upto a given accuracy.
This differs from our objective where we assume access to a noisy, unbiased first-order oracle and seek to compress the oracle output to the minimum number of bits, without losing the uncompressed convergence guarantees.


\mnote{At HT: You had added the references 17, 18. Just glancing through these papers, I am not sure if they have anything to do with $\ell_p$ oracles; [16], added now, does. Perhaps we should remove 17,18 ?}

\textbf{Notation}: Throughout the paper $q$ will denote the H\"{o}lder
conjugate of $p$ (that is, $q = p/(p-1)$).  $a \vee
b$ and $a \wedge b$ denote the $\max\{a,b\}$ and $\min\{a, b\}$, respectively. We
denote by $\log (\cdot)$ 
logarithm to the base $2$ and by $\ln (\cdot)$ logarithm to the base
$e$. $\ln^*(a)$ denotes  the number of times $\ln$ must be iteratively
 applied to $a$ before the result is less than or equal to $1$. $\{e_1, ..., e_d\}$ denotes 
  the standard basis.
We keep  notation consistent with our earlier work~\cite{mayekarratq}.

\section{Preliminaries}
\subsection{The Setup}
We extend the formulation of~\cite{mayekarratq} for Euclidean space to
general $\ell_p$ spaces. Formally, we consider the
problem of minimizing an unknown convex function $\displaystyle{f:\X
  \rightarrow \R}$ using {oracle access} to noisy subgradients of the
function ($cf.$~\cite{nemirovsky1983problem,bubeck2015convex}).  We
assume that the function $f$ is convex over the compact, convex domain
$\X$ such that $\sup_{x,y \in \X}\norm{x-y}_p \leq D$; we denote the
set of all such sets $\X$ by $\mathbb{X}$.  For a query point $x\in
\X$, the oracle outputs random estimates of the subgradient
$\hat{g}(x)$ which for all $x \in \X$ satisfy\footnote{Assumptions
  \eqref{e:asmp_unbiasedness} and \eqref{e:asmp_as_bound} imply that
  any function $f$ in $\mathcal{O}_{0,p}$ is Lipschitz continuous in
  the $\ell_p$ norm over the domain $\X$.}
\begin{align}
\label{e:asmp_unbiasedness}
\E{\hat{g}(x)|x} &\in \partial f(x), \\
\label{e:asmp_as_bound}
P(\norm{\hat{g}(x)}_q^2 &\leq B^2|x)=1,
\end{align}
where $q$ is the H\"older conjugate of $p$ and $\partial f(x)$ denotes
the set of subgradients of $f$ at $x$.
\begin{defn}[Almost surely bounded oracles]
A first order oracle which upon a query $x$ outputs only the
subgradient estimate $\hat{g}(x)$ satisfying the assumptions
\eqref{e:asmp_unbiasedness} and \eqref{e:asmp_as_bound} is termed an
almost surely bounded oracle. We denote the class of convex functions
and oracles satisfying assumptions \eqref{e:asmp_unbiasedness} and
\eqref{e:asmp_as_bound} by $\oO_{0,p}$.
\end{defn}
\begin{rem}[Mean square bounded oracles]\label{r:ms-oracles}
We remark that in the classic literature a more general oracle model
has been considered ($cf.$~ \cite{nemirovsky1983problem},
\cite{nemirovski1995information}). 
Specifically, \eqref{e:asmp_as_bound} is replaced by
  $\E{\norm{\hat{g}(x)}_q^2|x} \leq B^2.$ In this paper we restrict ourselves to
the almost surely bounded oracles. Nonetheless, using the general recipe
in \cite {mayekarratq} for converting quantizers for the almost surely bounded oracles
to mean square bounded ones, we can design quantizers for the latter model as well.
\end{rem}

In our setting, the outputs of the oracle are passed through a
quantizer. An {\em $r$-bit quantizer} consists of randomized mappings
$(\Qenc, \Qdec)$ with the encoder mapping $\Qenc:\R^d\to\{0,1\}^{r}$
and the decoder mapping $\Qdec: \{0,1\}^r\to \R^d$. The 
overall quantizer is given by the composition mapping $Q=\Qdec\circ
\Qenc$. Let $\Q_r$ be the set of all such $r$-bit quantizers.     

For an oracle $(f, O)\in \oO_{0, p}$ and an $r$-bit quantizer $Q$, 
let $QO= Q\circ O$ denote the 
composition oracle that outputs 
$Q(\hat{g}(x))$ for each query $x$. Let $\pi$ be an
algorithm with at most $T$ iterations with oracle access to
$QO$. We will call such an algorithm an {optimization  protocol}.  
Denote by $\Pi_T$ the set of all such optimization protocols with $T$
iterations. 
\begin{rem}[Memoryless, fixed length quantizers]
We note that the quantizers in $\Q_r$ are {memoryless} as well as {fixed length}
quantizers. That is, each new subgradient estimate at time $t$ will
be quantized without using any information from the previous updates to a fixed length code of $r$ bits.
\end{rem} 

Denoting the combined optimization protocol with its oracle $QO$ by 
$\pi^{QO}$ and the 
associated output as $x^*(\pi^{QO})$, we measure the performance of
such an optimization protocol for a given $(f,O)$ using the metric
$\ep_0(f, \pi^{QO}, p)$ defined as  $\ep_0(f, \pi^{QO}, p) \eqdef \E{f(x^*(\pi^{QO}))-\min_{x\in \X} f(x)}$.

Before proceeding, we recall the results for the case $r=\infty$.
These bounds will serve as a basic benchmark for our problem.
\begin{thm}\label{t:e_infty}
There exist absolute constants $c_0$ and $c_1$ where $c_1 \geq c_0 >0$ such that the following hold:
\begin{enumerate}
\item For $ p \geq 2$,
\[\displaystyle{
  \frac{c_1d^{1/2 -1/p} DB}{\sqrt{T}} \geq  \mathcal{E}_0^*{(T,\infty,p)}
\geq  \frac{c_0 d^{1/2 -1/p} DB}{\sqrt{T}};}\]
\item for\footnote{For certain range of $p$ closer to 2 the $\sqrt{\log d}$ factor can be removed; for simplicity, we state the slightly weaker result.} $2 >p \geq 1$,
\[\displaystyle{\frac{c_1  DB \sqrt{\log d}}{\sqrt{T}}  \geq
\mathcal{E}_0^*{(T,\infty,p)}
\geq   \frac{c_0DB }{\sqrt{T}}.}
\]
\end{enumerate}
\end{thm}
\noindent The lower bounds and the upper bounds can be found, for instance, in \cite[Theorem 1]{agarwal2012information} and \cite[Appendix C]{agarwal2012information}.

\begin{rem}\label{r:Motivationfoalphapin12}
An optimal achievable scheme for $p \in [1, 2)$ is the stochastic
  mirror descent with the mirror maps 
  $\norm{x}_{p^\prime}^2/(p^\prime-1)$, where $p^\prime$ is
  chosen appropriately for a given $p$. These algorithms require only
  that the expected squared $\ell_q$ norm of the gradient estimates are
  bounded.
\end{rem}
\begin{rem}\label{r:Motivationfoalphap}
An optimal achievable scheme for $p$ greater than $2$
  is simply projected subgradient descent(PSGD). To see this, note
  that PSGD gives a guarantee of ${D^{\prime}B^{\prime}}/{\sqrt{T}}$
  ($cf.$~\cite{nemirovski1995information}), where $D^{\prime}$ is the
  $\ell_2$ diameter and $B^{\prime}$ is the bound on
  $\E{\norm{\hat{g}}_2^2}$. Using the monotonicity of $\ell_q$ norms
  in $q$,  for $q \geq 2$ we have
  $\E{\norm{\hat{g}}_2^2} \leq \E{\norm{\hat{g}}_q^2} \leq B^2$. Also, the
  $\ell_2$ diameter of a unit $\ell_p$ ball is
  $d^{1/2-1/p}$. It follows that PSGD attains the upper bounds
in Theorem~\ref{t:e_infty}.
\end{rem}

The fundamental quantity of interest in this work is $r^*(T, p)$, the
minimum precision to achieve the optimization accuracy $\mathcal{U}(T,
p)$, the benchmark above from the classic setting. Specifically, we define
\begin{align}\label{e:r*} 
r^*(T, p)
&\eqdef \inf \{r \in \N: \mathcal{E}_0^*{(T,r, p)}  \leq
\mathcal{U}(T, p) \},    
 \end{align}
where\footnote{We take the  $\sup$ over all $X \in \mathbb{X}$ in the
  definition of  $\mathcal{E}_0^*{(T,r, p)}$, as it is unclear even in
  the infinite precision case if we can get matching upper bounds and
  lower bounds for any convex set in $\mathbb{X}$.} 
  \vspace{-0.15cm}
 \begin{align}\label{e:e_r}
  \mathcal{E}_0^*{(T,r, p)} &\eqdef \sup_{\X \in \mathbb{X}} \inf_{\pi \in \Pi_T}\inf_{Q \in
    \mathcal{Q}_r}\sup_{(f, O) \in \oO_{0,p}}\ep(f, \pi^{QO}, p),
  \\
  \nonumber
  \mathcal{U}(T, p) &:= \frac{4c_1d^{1/2 -1/p} DB}{\sqrt{T}}, \quad \forall p \in [2, \infty],\\ \nonumber
  \mathcal{U}(T, p) &:= \frac{4c_1 \sqrt{\log d} DB }{\sqrt{T}}, \quad \forall p \in [1, 2).
 \end{align}

\subsection{A Basic Convergence Bound for Quantized Gradients}
While our lower bounds hold for any kind of quantizers,
but we attain this bounds using unbiased quantizers. For 
such an unbiased quantizer $Q$, we characterize the performance in
terms of a parameter $\alpha_0(Q)$ defined earlier
in~\cite{mayekarratq}. We define

\[\displaystyle{\alpha_{0}(Q; p) := \sup_{Y \in \R^d: \norm{Y}_q^2\leq B^2
  \text{ a.s.}} 
\sqrt{\E{\norm{Q(Y)}_2^2}}, \quad  p \in [2, \infty]}, \]

\[\displaystyle{\alpha_{0}(Q; p) := \sup_{Y \in \R^d: \norm{Y}_q^2\leq B^2 \text{ a.s.}}
      \sqrt{\E{\norm{Q(Y)}_q^2}}, \quad p\in [1,2).}\]
      
Note that for all $p\geq 1$, the composed oracle $QO$ 
satisfies assumption \eqref{e:asmp_unbiasedness}.
Moreover, in view of Remarks~\ref{r:Motivationfoalphapin12}
and~\ref{r:Motivationfoalphap}, we have the following convergence 
guarantees for first-order stochastic optimization using gradients
quantized by $Q$. 

\begin{thm}\label{t:e_uq}
Consider a quantizer $Q$ for the gradients. There exists algorithms
$\pi \in \Pi_T$ which when used with oracle outputs quantized by $Q$
performs as follows.
\begin{enumerate}
\item For $ p \geq 2$,
\[\displaystyle{
  \frac{c_1d^{1/2 -1/p} D\alpha_{0}(Q; p)}{\sqrt{T}}\geq \sup_{(f, O)
    \in \oO_{0,p}}\ep(f,\pi^{QO}, p);}
  \]

\item for $2 > p \geq 1$,
\[
\displaystyle{\frac{c_1\sqrt{\log d}D \alpha_{0}(Q; p)}{\sqrt{T}} \geq \sup_{(f, O) \in \oO_{0,p}}\ep(f,\pi^{QO}, p)
.}
\]
\end{enumerate}
\end{thm}
In the rest of the paper, we design unbiased, fixed length
quantizers which have $\alpha_0( \cdot;p)$ of the same order as
$B$. Then, using Theorem \ref{t:e_uq}  the quantized updates give the
same convergence guarantees as that of the classical case, which 
leads to upper bounds for $r^*(T, p)$. Further, we derive lower bounds for
$r^*(T, p)$ to prove optimality of our quantizers.

An interesting insight offered by the result above, which is perhaps simple in
hindsight, is that even when dealing with $\ell_p$ oracles for $p\geq 2$, we only need to be concerned about the expected $\ell_2$ norm of
the quantizers output. It is this insight that leads to the realization that $SimQ^+$ is optimal for these settings. 


\section{Main Result: Characterization of $r^*(T, p)$}
   The main result of our paper is the almost complete
   characterization of $r^*(T,p)$. We divide the result into cases
   $p\in [1,2)$ and $p\geq 2$; as mentioned earlier, we use different
     quantizers for these two cases.
  \begin{thm}~\label{t:main}
    For stochastic optimization using $T$ accesses to a first-order
    oracle, the following bounds for $r^*(T,p)$ hold.
\begin{enumerate}
\item \newest{ For $p\geq 2$, we have \eq{ d^{2/p}\log &\,(2e\cdot d^{1-2/p}+2e) 
  \geq r^*(T, p) \geq \left(\frac{c_0}{4c_1}\cdot d^{1/p}\right)^2 \vee 
2 \log  \left(\frac{c_0}{4c_1}\cdot d^{1/2}\right) .} } 
  \item For $2 > p \geq 1,$ we have \eq{d\left(
\ceil{\log(2\sqrt{2} {\Delta_1}^{1/q}+2)}+3\right) +\Delta_2   \geq ~ r^*(T,
    p) \geq \left(\frac{c_0}{4c_1 \sqrt{\log d}}\right)^2 \cdot d,}
    where
$\Delta_1 =\ceil{\log\left(2+\sqrt{  18 + 6\ln \Delta_2} \cdot d^{1/2-1/q}\right)}$ and $\Delta_2 =\ceil{\log (1+\ln^*({d}/{3})) }.$
\end{enumerate}
\end{thm} 
\noindent Note that for $p \geq 2$ the upper bounds and lower bounds for  $r^*(T,p)$ are off by nominal factor of $\log (d^{1-2/p}+1)$. Also, for $p \in [1,  2)$ the bounds are roughly off by $O(\log d \cdot\log (\log d^{1/2-1/q})^{1/q})$ (ignoring the $\log^*d$ terms).

  We present the quantizers achieving these upper bounds, and the
  proof of the upper bounds, in the next two sections. For $p \geq 2$,
  we use a quantizer $SimQ$ and its extension $SimQ^+$, presented in
  Section~\ref{s:SimQ}. For $p\in [1,2)$, we use a combination of
    uniform quantization and the quantizer RATQ
    from~\cite{mayekarratq}, presented in
   Section~\ref{s:p-less-than-two}.
 

   Our lower bound proof is based on small modification of existing
   proofs. But an interesting element is that constructions that yield
   trivial bounds for convergence rate yield tight bounds when
   information constraints such as constraints on gradient precision
   are placed. The proof is deferred to Section~\ref{ss:lb-proof}.

We highlight the most interesting features of the result above in
separate remarks below.  
\begin{rem}[$r^*(T, p) $ is independent of $T$]
Theorem \ref{t:main} shows that $ r^*(T, p)$ is a function only of $p$
and $d$, and is
independent of $T$.  The number of queries $T$ is a proxy for the desired
optimization accuracy. Therefore, the fact that $r^*(T, p)$ is
independent of such a parameter is interesting.  We note, however,
that for oracle models with milder assumptions, such as mean square
bounded oracles, this may not hold. In fact, the results of~\cite{mayekarratq} suggest that for mean square bounded oracles $r^*(T,2)$ is dependent on $T$.
\end{rem}

\begin{rem}[Optimality for $p=\infty$]
  Our bounds match for $p=\infty$, namely our quantizer $SimQ$ offers optimal
  convergence rate with gradient updates at the least precision. A
  surprising observation is that this precision is merely $O(\log d)$,
  much smaller than $O(d)$ bits needed to recover the gradient vector
  under any reasonable loss function.
 \end{rem}

\begin{rem}[Optimality for $p=2$]
The lower bound for the case when $p=2$ already appeared in  \cite{mayekarratq},
giving $r^*(T, 2) = \Omega(d)$. Both \cite{alistarh2017qsgd} and
\cite{suresh2017distributed} give variable-length quantization schemes
to achieve this lower bound, but the worst-case precision can be
order-wise greater than $d$. Our recent work \cite{mayekarratq}
proposed a quantizer termed RATQ that was within a small factor of
$O(\log\log \log \ln ^*d)$ of this lower bound. Current work removes
this nagging factor using a different quantizer $SimQ^+$.
\end{rem}

\begin{rem}[Fixed precision]
The quantizer RATQ in~\cite{mayekarratq} remains
optimal upto a factor of $O(\sqrt{\log \ln^* d})$ for the more general problem of  characterizing $\mathcal{E}_0^*(T, r, 2)$
for any precision $r$ less than $d$ bits. In this
setting of small precision, the performance of $SimQ^+$ is much worse.
%
\end{rem}

\section{Our quantizers for $p \geq 2$}\label{s:SimQ}

We present our quantizer $SimQ$ and its extension $SimQ^+$. The former
is seen to be optimal for $p=\infty$ while the latter for $p=2$.

\subsection{An optimal quantizer for $p=\infty$}\label{s:infty}

\vspace{-.25cm}
\begin{figure}[ht]
\centering
\begin{tikzpicture}[scale=1, every node/.style={scale=1}]
\node[draw, text width= 8 cm, text height=,] {%
\begin{varwidth}{\linewidth}
            
            \algrenewcommand\algorithmicindent{0.7em}
            \renewcommand{\thealgorithm}{}
\begin{algorithmic}[1]
\Require Input $Y\in \R^d$, Parameter $B$ \State $i^*=\begin{cases}
i\quad w.p. \quad |Y(i)|/B\\ 0 \quad w.p. \quad 1-\norm{Y}_1/B
 \end{cases} $

\If { $i^* \in [d]$} \Statex \hspace{1cm} $j^*=sign(Y(i^*))$

\Else \Statex \hspace{1cm} ~$j^*=1$

  \EndIf \State \textbf{Output:} $\Qenc_{{\tt SimQ}}(Y; B) =i^*\cdot
  j^* $ \
\end{algorithmic}
\end{varwidth}};
 \end{tikzpicture}
 \caption{Encoder $\Qenc_{{\tt SimQ}}(Y; B)$ for
   $SimQ$}\label{a:E_SIMQ}
 \end{figure}

\vspace{-0.5cm}
 \begin{figure}[ht]
\centering
\begin{tikzpicture}[scale=1, every node/.style={scale=1}]
\node[draw, text width= 8 cm, text height=,] {%
\begin{varwidth}{\linewidth}
            \algrenewcommand\algorithmicindent{0.7em}
            \renewcommand{\thealgorithm}{}
            \renewcommand{\thealgorithm}{}
\begin{algorithmic}[1]
    \Require Input $i^\prime \in \{-d, -(d-1), \cdots 0,\cdots, d\} $
    \If { $i^\prime \neq 0$} \Statex \hspace{1cm} $Z=B
    sign(i^{\prime}) e_{|i^{\prime}|} $

\Else \Statex \hspace{1cm} ~$Z=0$

  \EndIf

\State \textbf{Output:} $\Qdec_{{\tt SimQ}}( i^\prime; B)=Z$
\end{algorithmic}
\end{varwidth}};
 \end{tikzpicture}
 \renewcommand{\figurename}{Algorithm}
 \caption{Decoder $\Qdec_{{\tt SimQ}}(i^\prime; B)$ for
   $SimQ$}\label{a:D_SIMQ} 
 \end{figure}
 \paragraph*{\textbf{ Simplex Quantizer ($\textbf{SimQ}$)}}
Our first quantizer $SimQ$ is described in Algorithms~\ref{a:E_SIMQ}
and~\ref{a:D_SIMQ}.  For $p=\infty$, our quantizer's input vector $Y$
is an unbiased estimate of the subgradient of the function at the
point queried and satisfies $\norm{Y}_1 \leq B$. $SimQ$ takes such a
$Y$ as an input and produces an output vector which, too, satisfies
both these properties. The main idea behind $SimQ$ is the fact that
any point inside the unit $\ell_1$ ball can be represented as a convex
combination of at the most $2d$ points: $\{e_i,-e_i: i \in [d],j \in
\{-1, 1\}\}$. With this observation, we use an $\ell_1$ sampling
procedure to obtain an unbiased estimate $i^*.j^*$ of the vector.
  At the decoder, upon observing $i^\prime=i^*.j^*$, the decoder
  simply declares $B j^*e_{i^*}$.


  \begin{thm}\label{t:SimQ}
 Let $Q$ be the quantizer $SimQ$ described in Algorithms \ref{a:E_SIMQ}, \ref{a:D_SIMQ}. Then, for 
  $Y$ such that $\norm{Y}_1 \leq B~a.s.$,
   $Q(Y)$ can be represented in $\log (2d+1)$ bits, $\E{Q(Y)|Y}=Y$, and
   $\alpha_0(Q, \infty)\leq B$.
 \end{thm}
 \begin{proof}
Since $i^*\in[d]$ and $j^*\in\{-1,1\}$, we can represent the output of
the encoder of $SimQ$ using $\log (2d+1)$ bits. Next, denoting
the quantizer $SimQ$ by $Q$, note that
\[
\E{Q(Y)|Y }= \sum_{i=1}^d B \cdot sign(Y(i))\cdot e_i \cdot
\frac{|Y(i)|}{B} = Y,
\]
namely $SimQ$ is unbiased. To complete the proof, note that
$\norm{Q(Y)}_2^2 \leq B^2~a.s.$.
\end{proof}
 Theorem \ref{t:SimQ} along with Theorem \ref{t:e_uq} establishes 
Theorem~\ref{t:main} for $p=\infty$.

\subsection{Our Quantizer for $p \in [2, \infty)$}
For this case,
    we need to quantize inputs that are bounded in
    $\ell_q$ norm with $q\in (1,2]$ so that the quantized output is unbiased and has small expected $\ell_2$ norm square; we will use $SimQ^+$ to do this.

\paragraph*{$\textbf{SimQ}^+$}
The quantizer $SimQ^+$ outputs the average of $k$ independent repetitions of the $SimQ$
quantizer for a given input vector. The input vectors $Y$ satisfy
 $\norm{Y}_1\leq B d^{1/p}$. Therefore, we use $SimQ$ with
parameter $B d^{1/p}$ instead of $B$. The repetitions help reduce the
error to compensate for the extra loss factor. Specifically, the output of
$SimQ^+$ denoted by  $Q(Y)$ is given by
 \begin{align}\label{e:SIMQ+}
 Q(Y) =\frac{1}{k}\cdot\sum_{i =1}^{k} Q^i_{\tt SimQ}(Y; Bd^{1/p}),
  \end{align}
 where $Q^i_{\tt SimQ}$ are independent iterations of $SimQ$.


The next component of $SimQ^+$ is how the encoder of $SimQ^+$ expresses the output of these $k$ copies of $SimQ$ to attain compression. If represented naively, this
will require $O(d^{2/p}\log d)$. But we can do much better since we only need the average value of these entries. For that, we can simply
report the {\em type} of this vector -- the frequency of each
 index in the $k$ length sequence. The signs of the
 input coordinates for the non-zero entries can be sent separately.

 Note that there are $d+1$ indices
 overall, as $SimQ$ can pick any index from $\{0, \ldots d\}$.
 Therefore, the total number of types is ${{d+k}\choose{k}}$, which
 can at the most be $(\frac{ed+ek}{k})^k$ bits. Hence, the precision
 needed to represent the type is at the most $k \log e +
 k\log(\frac{d}{k}+1)$.

 The type of the input can be used to determine a set $\mathcal{I}_0$ of non-zero indices
 that appear at least once. There are at most $k$ such entries. Therefore, we can use
 a binary vector of length $k$ to store the signs for these entries. We use this representation in $SimQ^+$, with the indices in $\mathcal{I}_0$ represented in the vector in increasing order.
%
 

\begin{thm}\label{t:p2infty}
 For a $p \in [2, \infty)$, let $Q$ be the quantizer $SimQ^+$ described in \eqref{e:SIMQ+}. Then, for 
  $Y$ such that $\norm{Y}_q \leq B~a.s.,$
   $Q(Y)$ can be represented in $k \log e + k\log(\frac{d}{k}+1)+k$ bits, $\E{Q(Y)|Y}=Y$, and
   $\displaystyle{\alpha_0(Q;p) \leq \sqrt{\frac{B^2 d^{2/p}}{k}+B^2}}$.
 \end{thm}
\begin{proof}
  We already saw how to represent the output of the $k$ copies of $SimQ$
using $k \log e + k\log(\frac{d}{k}+1)+k$ bits. 
For bounding $\alpha_0(Q;p)$, note from~\eqref{e:SIMQ+}
that $SimQ^+$ is an unbiased quantizer since $SimQ$ is unbiased.
Further, denoting by $Q_i(Y)$ the output $Q^i_{\tt SimQ}(Y;
Bd^{1/p})$, we get
\eq{
  \E{\norm{Q(Y)}_2^2}&=\E{\norm{Q(Y)-Y}_2^2}+\E{\norm{Y}_2^2}\\ &=\frac{1}{k^2} \sum_{i=1}^{k}\E{\E{\norm{Q_i(Y)-Y}_2^2|Y}}+\E{\norm{Y}_2^2}\\ &= \frac{\E{\norm{Q_1(Y)-Y}_2^2}}{k}+\E{\norm{Y}_2^2}
\\
  &\leq
  \frac{d^{2/p}B^2}{k}+B^2,
} where the first identity uses the fact
that $Q(Y)$ is an unbiased estimate of $Y$; the second uses the fact
that $Q_i(Y)-Y$ are zero-mean, independent random variables when conditioned on $Y$;
the third uses the fact that $Q_i(Y)-Y$ are identically distributed;
and the final inequality is by the performance of $SimQ$.
\end{proof}

The proof of upper bound for $p\in [2,\infty)$ in Theorem~\ref{t:main} is completed by setting $k=d^{2/p}$ and using Theorems~\ref{t:p2infty} and~\ref{t:e_uq}.

{
\section{Our Quantizers for $p \in [1,2)$}\label{s:p-less-than-two}
For p in $[1,2)$, the oracle yields unbiased subgradient estimates
  $Y$ such that $\norm{Y}_q \leq B$ almost surely. Our goal is to
  quantize such $Y$s in an unbiased manner and ensure that
  $\E{\norm{Q(Y)}_q^2}$ is $O(B^2)$. It can be seen that a simple unbiased
  uniform quantizer will achieve this using $d (\log d )^{1/2 -1/q}$.}
  However, our goal here is to get a result that is stronger than this
  baseline performance. To that end, we split the input vector $Y$ in
  two parts $Y_1$ and $Y_2$ with the first part having $\ell_\infty$
  norm less than $c$ and the second part having less than $d/\Delta_1$
  nonzero coordinates. We use an ``$\ell_\infty$ ball quantizer'' (a
  uniform quantizer) for $Y_1$ and an ``$\ell_2$ ball quantizer'' for
  $Y_2$.

Specifically, set $\displaystyle{c :=\frac{B
    \Delta_1^{1/q}}{d^{1/q}}},$ where $\Delta_1$ is that in Theorem
\ref{t:main}.  Then, define
\begin{align}\label{e:Y1}
Y_1 :=\sum _{i=1}^{d} Y(i)\indic{\{|Y(i)| \leq c\}} e_i,\text{~~} Y_2
:=\sum _{i=1}^{d} Y(i)\indic{\{|Y(i)| > c\}} e_i.
\end{align}
Clearly, $\norm{Y_1}_\infty\leq c$. Further, since $\norm{Y}_q \leq
B$, the number of nonzero coordinates in $Y_2$ can be at the most
$B^q/c^q=d/\Delta_1$. For quantizing $Y_1$, we use the coordinate-wise
uniform quantizer (CUQ) described below.

\paragraph*{\textbf{CUQ}}
We note that $CUQ$ is an unbiased, randomized uniform quantizer which
has appeared in multiple works recently
($cf.$~\cite{suresh2017distributed},~\cite{mayekarratq}). We follow
the treatment in \cite{mayekarratq}.  CUQ has two parameters: $M$
which describes the complete dynamic range $[-M, M]$ of the quantizer;
$k$ which describes the number of levels to which $[-M, M]$ is
quantized to uniformly.  In order to quantize $Y_1$ in \eqref{e:Y1},
we set
\begin{align}\label{e:CUQparam}
M = c, \quad \log(k+1)=\ceil{\log(2\sqrt{2}\Delta_1^{1/q}+2)}.
\end{align} 

\begin{lem}\label{l1}
Let $Q_{\tt u}$ be the quantizer CUQ with parameters $M$ and $k$ set
as in \eqref{e:CUQparam}. Then, for $Y$ such that $\norm{Y}_q \leq
B~\text{a.s.}$ and $Y_1$ as that in \eqref{e:Y1}, $Q_{\tt u}(Y_1)$ can
be represented in $ d \ceil{\log(2\sqrt{2}\Delta_1^{1/q}+2)} $ bits,
$\E{Q_{\tt u}(Y_1)\mid|Y}=Y_1$, and $\E{\norm{Q_{\tt u}(Y_1)}_q^2}\leq
3B^2$.
\end{lem}
\begin{proof}
 CUQ requires a precision of $d\log(k+1)$, which coincides with the
 statement above for our choice of $k$.  To see unbiasedness, note
 that CUQ is an unbiased quantizer as long as all the coordinates of
 the input do not exceed $M$. Since we have set $M=c$ and
 $\norm{Y_1}_\infty=c$, this property holds.  Finally, to show that
 $\E{\norm{Q_{\tt u}(Y_1)}_q^2}\leq 3B^2$, note that $\E{\norm{Q_{\tt
       u}(Y_1)}_q^2} \leq 2\E{\norm{Q_{\tt u}(Y_1)
     -Y_1}_q^2}+2\E{\norm{Y_1}_q^2}.$ Also,
\vspace{-0.5cm} \eq{ \E{\norm{Q_{\tt u}(Y_1) -Y_1}_q^2} &\leq
  \E{\sum_{i \in [d]}|Q_{\tt u}(Y_1(i)) -Y_1(i)|^q}^{2/q}\\&\leq B^2,}
where the first inequality uses Jensen's inequality as $(\cdot
)^{2/q}$ is a concave function and second follows from fact that for
$M$ set as in \eqref{e:CUQparam} we have that $|Q_{\tt
  u}(Y_1)(i)-Y_1(i)| \leq \frac{2M}{(k-1)}$ a.s., $\forall i \in [d]$,
by the description of CUQ.
\end{proof}
In order to quantize $Y_2$, we indicate the coordinates with non-zero
entries. This takes less than $d$ bits.  Then, we quantize the
restriction $Y_2^\prime$ of $Y_2$ to these nonzero entries. Recall
that the dimension of $Y_2^{\prime}$ is less than $d^{\prime}
:={d}/{\Delta_1}$. Also, the $\ell_2$ norm of $Y_2^{\prime}$ is less
than $\norm{Y}_2\leq\norm{Y}_q d^{1/2-1/q} \leq B d^{1/2-1/q} =:
B^{\prime}$.

We need a quantizer $Q$ such that $\E{\norm{Q(Y_2^{\prime})}_q^2}$ is
$O(B^2)$.  As seen in the proof of Lemma~\ref{l1}, one way to do this
is to ensure $\E{\norm{Q(Y_2^{\prime})-Y^{\prime}}_q^2}$ is $O(B^2), $
which, in turn, can be ensured if
$\E{\norm{Q(Y_2^{\prime})-Y_2^{\prime}}_2^2}$ is $O(B^2)$.  To achieve
this, we can use an unbiased quantizer for the unit $\ell_2$ ball in
$\R^d$, which can quantize it to an MSE of $O( 1/d^{1-2/q})$ using
$O(d \log (d^{1/2-1/q})$ bits.
We note that $SimQ^+$, while optimal for the stochastic optimization
use-case, does not yield the required scaling of bits in MSE.  A
candidate quantizer is RATQ of \cite{mayekarratq}, which is, in fact,
close to information theoretically optimal.

\newest{We note that RATQ is a quantizer used to quantize random
  vectors in $\R^d$ which have Euclidean norm almost surely bounded by
  $B$. The optimal parameters for RATQ are set in terms of $B$ and
  $d$. Since $Y_2^{\prime}$ has the Euclidean norm  almost surely bounded by
  ${B^\prime}=B d^{1/2-1/q}$ and its dimension is
  ${d^\prime}={d}/{\Delta_1}$, we can 
  set the parameters of RATQ in terms of
 $B^{\prime}$ and $d^{\prime}$. 
 We set 
\begin{align}\label{e:param_RATQ}
&\nonumber
m=\frac{3{B^{\prime}}^2}{d^{\prime}}, \quad m_0=\frac{2{B^{\prime}}^2 }{d^{\prime}} \cdot  \ln s, \quad \log h=\ceil{\log(1+\ln^\ast(d^{\prime}/3))},
\\
& s=\log h, \quad \log(k+1) = \Delta_1.
\end{align}}
%


\begin{lem}\label{l2}
Let $Q_{{\tt at}, R}$ be the quantizer RATQ with parameters set as \eqref{e:param_RATQ}.
Then, for $Y$ such that $\norm{Y}_q \leq B ~{a.s.}$  and $Y_2^{\prime}$ the restriction of $Y_2$ in \eqref{e:Y1},    $Q_{{\tt at}, R}(Y_2^{\prime})$ can be represented in $2d+\Delta_2$ bits, $\E{Q_{{\tt at}, R}(Y_2^{\prime})\mid|Y}=Y_2^{\prime}$, and $\E{\norm{Q_{{\tt at}, R}(Y_2^{\prime})}_q^2}\leq 3B^2.$
\end{lem}
\begin{proof}

First, we note that the output of RATQ can be represented in $\ceil{d^{\prime}/s}(\log h) +d \log (k+1)$ bits, which, in this case, is less than
\eq{
\frac{d}{\Delta_1 \log h} \cdot \left(\log h \right)+ \log h+\left(\frac{d}{\Delta_1 } \log(k+1)\right) \leq 2d+\Delta_2.
}
For unbiasedness, note that for our choice of $m, m_0, h$, RATQ is always an unbiased quantizer of the input.
\newest{Finally, for showing $\E{\norm{Q_{{\tt at}, R}(Y_2^{\prime})}_q^2}\leq 3B^2$, 
we note that
\eq{\E{\norm{Q_{{\tt at}, R}(Y_2^{\prime})}_q^2}
  &\leq 2\E{\norm{Q_{{\tt at}, R}(Y_2^{\prime})-Y_2^{\prime}}_q^2}
  +2\E{\norm{Y_2^{\prime}}_q^2}
\\ 
  &\leq 2\E{\norm{Q_{{\tt at}, R}(Y_2^{\prime})-Y_2^{\prime}}_q^2}
  +2B^2
\\ 
  &\leq 2\E{\norm{Q_{{\tt at}, R}(Y_2^{\prime})-Y_2^{\prime}}_2^2}
  +2B^2.
}
The proof will be complete upon showing that $\E{\norm{Q_{{\tt at},
      R}(Y_2^{\prime})-Y_2^{\prime}}_2^2} \leq B^2/2$, towards
which we apply~\cite[Lemma 6.3]{mayekarratq} to get
\[
\E{\norm{Q_{{\tt at}, R}(Y_2^{\prime})-Y_2^{\prime}}_2^2}
\leq B^2 d^{1-2/q} \cdot  \frac{9+3\ln s}{(k-1)^2},
\]
and substituting our choice of $k$.}
\end{proof}

The overall quantizer $Q$ of input vector $Y$ is the sum of quantized
outputs of $Y_1$ and $Y_2$. By Lemmas \ref{l1} and \ref{l2}, the
quantized output of $Y$ can be represented in $d\left(
\ceil{\log(2\sqrt{2} {\Delta_1}^{1/q}+2)}+3\right) +\Delta_2$ bits\footnote{This
accounts for the communication needed to send the nonzero indices of
$Y_2$, too.}. Furthermore, $\alpha_0(Q;p) \leq \sqrt{12}B$. These facts
along with Theorem \ref{t:e_uq} prove the upper bound in
Theorem~\ref{t:main} for $p\in [1,2)$.

\section{Proof of lower bounds in Theorem \ref{t:main}}\label{ss:lb-proof} 
\newest{We derive the lower bounds on the maxmin error $ \mathcal{E}_0^*{(T, r,p)}$ defined in \eqref{e:e_r}. The lower bounds of Theorem~\ref{t:main} will follow by upper-bounding 
  $\mathcal{E}_0^*{(T, r,p)}$ by $\mathcal{U} (T,p)$.}
These information theoretic lower bounds generalize the
ones derived for the Euclidean case in \cite{mayekarratq}.

In our proof below, we
 use the reduction of convex optimization to mean
 estimation given in~\cite{agarwal2012information}.
 We use two different types of oracle constructions. The first one is based on
   Bernoulli product distribution and yields the 
lower bound using a strong data processing inequality
 from \cite{duchi2014optimality}. 
The second one is  Paninski's construction \cite{paninski2008coincidence} and uses a 
 strong data processing 
   inequality type bound ($cf.$ the chi-square contraction bound in~\cite{ACT:18}). 
Interestingly, for $p > 2 $ we
   need  to use both the oracle constructions, whereas for $p \in [1,
     2)$ Bernoulli product construction is sufficient. Heuristically,
     for $p > 2$,  Paninski's construction captures the difficulty for
     optimization but does not pose much additional difficulty 
for quantized oracles. On the other hand, the Bernoulli product
     construction is not the bottleneck for optimization but poses
a significant challenge  if the oracle is quantized. 

\begin{thm}\label{t:lb1}
For $p \in [2, \infty]$, we have
\[
 \mathcal{E}_0^*{(T, r,p)}
\geq 
\left( \frac{c_0DB\cdot d^{1/2-1/p}}{\sqrt{T}} \cdot \sqrt{\frac{d}{d \wedge 2^r}}\right) \vee \left(\frac{c_0DB}{\sqrt{T}} \cdot \sqrt{\frac{d}{d \wedge r}} \right)
 .\]
\end{thm}
\begin{proof}
~

\paragraph{First Lower Bound}
 For simplicity, we assume $\X=\{x:\norm{x}_\infty\leq D/(2d^{1/p}) \}$.
Let $\V\subset\{-1,1\}^d$ be the maximal $d/4$-packing in Hamming distance, namely it is a collection of vectors such that
any two vectors $\alpha, \alpha^\prime\in V$, $d_H(\alpha, \alpha^\prime)\geq d/4$. As is well-known, there exists such a packing of cardinality $2^{c_2 d}$, where $c_2$ is a constant.
Consider convex functions $f_\alpha$, $\alpha\in \V$, with domain $\X$ and satisfying assumptions \eqref{e:asmp_unbiasedness} and
\eqref{e:asmp_as_bound} given below:
 \[
f_{\alpha}(x):=\frac {2B\delta}{d} \sum_{i=1}^{d} \alpha(i) x(i).
\]
Note that the gradient of $f_\alpha(x)$ is given by $2B\delta\alpha/d$ for each $x\in \X$. 
For each $f_\alpha$, consider the corresponding gradient oracles $O_\alpha$ which outputs 
$e_i \cdot B$ and $-e_i \cdot B$ with probabilities $(1+2\delta\alpha(i))/2d$ and $(1-2\delta\alpha(i))/2d$, respectively. We denote the distribution of output of oracle $O_\alpha$ by $P_\alpha$.

Let $V$ be distributed uniformly over $\V$. Consider the multiple hypothesis testing problem of determining $V$ by observing
samples from $\Qenc(Y)$ with $Y$ distributed as $P_V$. Consider an optimization algorithm that outputs $x_T$ after $T$ iterations. Then, we have 

\begin{align*}
\E{f_\alpha(x_T)-f_\alpha(x^*)}&\geq \frac{DB\delta}{4d^{1/p}} P\left(f_\alpha(x_T)-f_\alpha(x^*)\geq\frac{DB\delta}{4d^{1/p}}\right)
\\
&{=} \frac{DB\delta}{4 d^{1/p}} P\left( \frac {B\delta}{d} \alpha^T(x_T-x^*) \geq\frac{DB\delta}{8d^{1/p}} \right)
\\
&{=} \frac{DB\delta}{4d^{1/p}} P\left( \frac {B\delta}{d} \norm{x_T-x^*}_{1} \geq\frac{DB\delta}{8d^{1/p}} \right)
\\
&{=}  \frac{DB\delta}{4d^{1/p}} P\left(\norm{(2d^{1/p}/D)x_T + \alpha}_1\geq \frac {d}{4}\right),
\end{align*}
where the second identity holds since
  $sign(\alpha(i))=sign(x_T-x^*)$ and the final identity is obtained by noting that the optimal value $x^*$ for $f_\alpha$ is $-(D/2d^{1/p})\alpha$.  Note that all $\alpha,
\alpha^\prime \in\V$ satisfy $\norm{\alpha-\alpha^\prime}_1\geq d/2
$. Consider the following test for the aforementioned hypothesis
testing problem. We execute the optimization protocol using oracle
$O_V$ and declare the unique $\alpha\in V$ such that
$\norm{(2d^{1/p}/D)x_T + \alpha}_1 < d/4$. The probability of error
for this test is bounded above by $P\left(\norm{(2d^{1/p}/D)x_T +
  \alpha}_1\geq \frac{d}{4}\right)$, whereby the previous bound and
Fano's inequality give
\[
\E{f_\alpha(x_T)-f_\alpha(x^*)}\geq \frac{DB\delta}{4d^{1/p}} \left(1-
\frac{TI(V; Q(Y))+1}{\log |\V|}\right).
\]
For a quantizer $Q$ with precision $r$, using, for instance, 
the chi-square contraction bound from~\cite{ACT:18}, we have $I(V; Q(Y))\leq 
\delta^2\min\{2^r, d\}/d$. Therefore,
\[
\max_{\alpha}\ep_0(f, \pi^{QO})\geq \frac{{DB}\delta}{4d^{1/p}} \bigg(1 - \frac{T\delta^2(\min\{2^r, d\}/d)}{c_2d}\bigg).
\]
The proof is completed by maximizing the right-side over~$\delta$.

\paragraph{Second Lower Bound}
We consider a slightly different family of convex functions still parametrized by $\alpha \in V$, with everything else remaining the same as the first bound.
Consider convex functions $f_\alpha$, $\alpha\in \V$, with domain $\X$ and satisfying assumptions \eqref{e:asmp_unbiasedness} and
\eqref{e:asmp_as_bound} given below:
 \[
f_{\alpha}(x):=\frac {2B\delta}{d^{1/q}} \sum_{i=1}^{d} \alpha(i) x(i).
\]
Note that the gradient of $f_\alpha(x)$ is given by $2B\delta\alpha/d^{1/q}$ for each $x\in \X$. 
For each $f_\alpha$, consider the corresponding gradient oracles $O_\alpha$ which outputs 
independent values for each coordinate,
with the value of $i$th coordinate taking values $B/d^{1/q}$ and $-B/d^{1/q}$ with probabilities $(1+2\delta\alpha(i))/2$ and $(1-2\delta\alpha(i))/2$, respectively. We denote the distribution of output of oracle $O_\alpha$ by $P_\alpha$.

By similar argument as in the first bound, we have
\newest{
\begin{align*}
\E{f_\alpha(x_T)-f_\alpha(x^*)}&\geq \frac{DB\delta}{4} P\left(f_\alpha(x_T)-f_\alpha(x^*)\geq\frac{DB\delta}{4}\right)
\\
&{=} \frac{DB\delta}{4} P\left( \frac {B\delta}{d^{1/q}} \alpha^T(x_T-x^*) \geq\frac{DB\delta}{8} \right)
\\
&{=} \frac{DB\delta}{4} P\left( \frac {B\delta}{d^{1/q}} \norm{x_T-x^*}_{1} \geq\frac{DB\delta}{8} \right)
\\
&{=}  \frac{DB\delta}{4} P\left(\norm{(2d^{1/p}/D)x_T + \alpha}_1\geq \frac {d}{4}\right),
\\&{=}\frac{DB\delta}{4} 
 \left(1-
\frac{TI(V; Q(Y))+1}{\log |\V|}\right)
\end{align*}
}
For a quantizer $Q$ with precision $r$, using the strong data
processing inequality bound from~\cite[Proposition
  2]{duchi2014optimality}, we have $I(V; Q(Y))\leq 360
\delta^2\min\{r, d\}$. Therefore,
\[
\max_{\alpha}\ep_0(f, \pi^{QO})\geq \frac{{DB}\delta}{4} \bigg(1-
\frac 1{c_2d} - \frac{360T\delta^2\min\{r, d\}}{c_2d}\bigg).
\]
The proof is completed by maximizing the right-side over $\delta$.

\end{proof}

\newest{
  Next, noting that the second lower bound in Theorem \ref{t:lb1} also holds for $p \in [1, 2)$, we obtain the following lower bound.
\begin{thm}
For $p \in [1, 2)$, we have
\[
 \mathcal{E}_0^*{(T, r,p)}
\geq 
\frac{c_0DB}{\sqrt{T}} \cdot \sqrt{\frac{d}{d \wedge r}} 
 .\]
\end{thm}
 }

\section{Comments on general tradeoff and mean square bounded oracles}
We close with the remark that an almost complete characterization of
$\mathcal{E}_0^*{(T,r, p)}$, for any $r, p$ can be obtained using our
quantizers and the ideas developed in this paper. In fact, our lower
bounds on $r^*(T, p)$ are derived via lower bounds  on
$\mathcal{E}_0^*{(T,r, p)}$ which hold for all $r,p$ (these can be found in the
extended version). For upper bounds when $p \in [1, 2)$,
  note that the parameter $k$ of $SimQ^+$ gives us a nice lever to operate
  under any precision constraint $r \geq \log d$. It turns out that
  such a quantizer leads to upper bounds which are off by at the most
by a  $\sqrt{\log d}$ factor. For upper bounds in the case of $p \in [1,
    2)$, note that classical sampling techniques such as uniform
    sampling without replacement maybe used to sample a subset of
    coordinates and quantize them using the quantizers described
    here. Even in this case the upper bound and lower bounds are off
    by a nominal factor of $\sqrt{ \log d \cdot \log \log
      d/{q}}$.
    However, we believe that removing these
      remaining factors can lead to new quantizers, and is of research
      interest.

For mean square bounded oracles mentioned in
        Remark~\ref{r:ms-oracles}, the bias in the quantized oracle
        output is nearly inevitable. In our previous work
        \cite{mayekarratq}, we proposed appropriate {\em{gain-shape}}
        quantizers for quantizing the oracle output in the Euclidean
        setup, which resulted in lesser bias over standard
        quantizers. This idea is valid for the general $\ell_p$ setup;
        in particular, we can use a gain quantizer to quantize the
        $\ell_q$ norm of the oracle output and a
        {{shape}} quantizer to quantize the oracle output vector normalized by the
        $\ell_q$ norm, the shape of the oracle output vector.
        Note that the shape vector has
        has $\ell_q$ norm $1$, which
        allows us to use the quantizers developed in this paper to
        quantize the shape. The gain is a scalar random variable which
        has its second moment bounded by $B^2$. To quantize such a
        random variable, we can use the quantizer proposed in
        \cite{mayekarratq} termed {\em{Adaptive Geometric Uniform
            Quantizer}} (AGUQ).
      \newest{ Clearly, the lower bounds
      for almost surely bounded oracles remain valid for mean square bounded oracles as well. Additionally, we can also derive lower bounds for a specific class
        of quantizers, such as those derived in~\cite{mayekarratq}, which help in capturing the reduction in the convergence rate due to mean square bounded noise. However, we do not have matching bounds, even for the Euclidean case.}


\section*{Acknowledgement} The authors would like to thank Arya Mazumdar and Anand Theertha Suresh for useful discussions. 

Prathamesh Mayekar is supported by a PhD
  fellowship from Wipro Limited. Himanshu Tyagi is supported by a grant from Robert Bosch
  Center for Cyberphysical Systems (RBCCPS), Indian Institute of
  Science, Bangalore and the grant EMR/2016/002569 from the Department
  of Science and Technology (DST), India. 
\bibliography{IEEEabrv,tit2018}
\bibliographystyle{IEEEtranS}

\end{document}